\documentclass[journal,transmag]{IEEEtran}
\usepackage{cite}
\usepackage{amssymb}
\usepackage{amsfonts}
\usepackage[cmex10]{amsmath}
\usepackage[pdftex]{graphicx}
\usepackage{array}
\usepackage{stfloats}
\usepackage{url}
\usepackage{epstopdf}
\usepackage{cases}
\usepackage{amsthm}
\allowdisplaybreaks

\newtheorem{theorem}{Theorem}
\newtheorem{lemma}{Lemma}

\newtheorem{remark}{Remark}

\newcommand{\RR}{\mathbb{R}}
\newcommand{\HD}{\mathcal{H}_{D}(\Omega)}
\newcommand{\HN}{\mathcal{H}_{N}(\Omega)}
\newcommand{\grad}{\nabla}
\newcommand{\curl}{\nabla\times}
\newcommand{\diver}{\nabla\cdot}
\newcommand{\jump}[1]{\left[\!\left[#1\right]\!\right]}
\newcommand{\bn}{\mathbf n}
\newcommand{\bF}{\mathbf F}
\newcommand{\bG}{\mathbf G}
\newcommand{\bE}{\mathbf E}
\newcommand{\bH}{\mathbf H}
\newcommand{\bN}{\mathbf N}
\newcommand{\bD}{\mathbf D}
\newcommand{\bA}{\mathbf A}

\begin{document}

\title{Harmonic Vector Fields and Betti Numbers in Bounded Three-Dimensional Electromagnetic Domains}
\author{\IEEEauthorblockN{Wei Jiang\IEEEauthorrefmark{1} and Jie Liu\IEEEauthorrefmark{2}~\IEEEmembership{Member,~IEEE}}
\IEEEauthorblockA{\IEEEauthorrefmark{1}School of Mechatronics Engineering, Guizhou Minzu University, Guiyang 550025, China}
\IEEEauthorblockA{\IEEEauthorrefmark{2}School of Mathematical Science, Guizhou Normal University, Guiyang 550025, China}
\thanks{Manuscript received XXXX XX, XXXX;
Corresponding author: Wei Jiang (email: jwmathphy@163.com).}}


\IEEEtitleabstractindextext{
\begin{abstract}
The topology of a bounded three-dimensional domain can strongly affect electromagnetic fields. In a topologically complex domain, a curl-free field may not have a globally single-valued scalar potential and a divergence-free field may also fail to have a global vector potential. These topological effects lead naturally to harmonic vector fields in the Helmholtz decomposition. This paper gives a geometric and constructive study of such fields using vector analysis, circulation integrals, cutting surfaces, scalar Laplace problems, Stokes' theorem, and Green's identity. The first Betti number is interpreted through independent handle-type circulations, whereas the second Betti number counts enclosed voids. Direct proofs are given for the dimensions of the Neumann and Dirichlet harmonic-field spaces. The analysis is also extended to anisotropic lossless media. The results provide a simple topological interpretation of harmonic fields and physical DC modes in bounded electromagnetic resonators.
\end{abstract}
}

\maketitle
\IEEEdisplaynontitleabstractindextext
\IEEEpeerreviewmaketitle

\section{Introduction}

\IEEEPARstart{M}{ost} electromagnetic textbooks introduce scalar and vector potentials, as well as the Helmholtz decomposition, without emphasizing the topology of the field domain. This is usually harmless in the whole space or in a topologically simple bounded region. The situation changes when a bounded region contains handles or enclosed voids. In such a domain, a curl-free field does not necessarily admit a globally single-valued scalar potential. Similarly, the existence of a global vector potential for a divergence-free field may require additional global conditions. These topological effects are closely related to harmonic vector fields.

The relation between vector fields and the topology of bounded three-dimensional domains has been studied from both mathematical and electromagnetic viewpoints. Vector-calculus and Hodge-type decompositions show that harmonic fields represent finite-dimensional components that cannot, in general, be removed by local differential conditions alone \cite{Cantarella2002,Amrouche1998}. In computational electromagnetism, harmonic Neumann and harmonic Dirichlet spaces arise naturally from the de Rham complex, and their dimensions are determined by the Betti numbers of the domain \cite{Hiptmair2002}. More systematic treatments of topology in electromagnetic field theory and computation can be found in \cite{Bossavit1998Book,GrossKotiuga2004}.

Potential formulations provide a direct manifestation of these topological effects. In multiply connected domains, a curl-free field may possess nonzero circulation around independent handles and therefore fail to admit a globally single-valued scalar potential. Cutting surfaces are commonly introduced to remove this obstruction and to represent the missing global information explicitly \cite{Kotiuga1989,Benedetti2012}. Homology and cohomology have also been used to describe the global degrees of freedom associated with electromagnetic boundary-value problems and magneto-quasistatic formulations \cite{Pellikka2010,DlotkoSpecogna2013}. Related topological structures also appear in electrostatic vector-potential formulations through second-cohomology generators \cite{Specogna2011}.

The same ideas arise naturally in finite-element formulations. Finite-element exterior calculus provides a general framework that preserves the differential-complex and cohomological structure of Maxwell equations \cite{Arnold2006}. The role of computational domain topology in scalar- and vector-potential formulations has also been studied systematically \cite{AlonsoValli2010}. Constructive finite-element bases of harmonic fields and of the first de Rham cohomology group have been developed for three-dimensional electromagnetic problems \cite{Alonso2013}. In multiply connected domains, circulation constraints along suitable cycles can be used to control the topological kernel of the curl operator \cite{Alonso2018}. Harmonic Dirichlet and Neumann fields also play an important role in Maxwell equations on perforated domains \cite{Qin2020}. More general harmonic Dirichlet--Neumann fields have been studied in a functional-analytic setting \cite{Pauly2022}.

Topological harmonic fields also have a direct physical interpretation in electromagnetic resonators. For resonators with complex geometric topology, mixed finite-element formulations have been developed to distinguish physical electromagnetic modes from nonphysical numerical modes \cite{Jiang2016}. Physical zero-frequency modes were later identified as genuine electromagnetic modes whose multiplicities depend on the topology of the resonator \cite{Jiang2019}. Numerical eigensolvers for anisotropic cavity resonators further showed the importance of separating physical modes from spurious zero modes by enforcing the divergence constraint \cite{Jiang2020}.

Harmonic vector fields are therefore well established in mathematical treatments of Maxwell equations and in topology-aware computational electromagnetics. However, their role is often encountered indirectly through scalar-potential cuts, global degrees of freedom, operator null spaces, or zero-frequency modes. The direct relation among harmonic vector fields, Betti numbers, and physical DC modes is consequently not always made explicit in the electromagnetics literature.

The purpose of this paper is not to claim that the classical Betti number dimension formulas are new. The aim is to give a direct and constructive electromagnetic interpretation of these formulas using ordinary vector analysis. The first and second Betti numbers are described geometrically through independent handles and enclosed voids. Constructive bases are then developed for the harmonic Neumann and harmonic Dirichlet spaces using circulation integrals, cutting surfaces, and scalar Laplace problems. Their dimensions are established using Stokes' theorem, Green's identity, and global potential arguments. The results are further extended to inhomogeneous anisotropic lossless media. Finally, the harmonic field dimensions are connected with physical DC modes and finite-element null spaces in electromagnetic resonators.

\section{Harmonic Vector Fields and Topology}
\indent In this section, we discuss harmonic vector fields in electromagnetics. The first and second Betti numbers are used to describe the topology of the domain. Finally, a basic lemma gives the condition for a curl-free field to have a globally single-valued scalar potential.

Let $\Omega\subset\RR^3$ be a bounded, orientable, and connected domain with piecewise smooth boundary. Only domains of finite topology are considered. Write
\begin{equation}\label{eq:boundary}
\partial\Omega=\Gamma_0\cup\Gamma_1\cup\cdots\cup\Gamma_M,
\qquad \Gamma_i\cap\Gamma_j=\varnothing\quad (i\ne j),
\end{equation}
where $\Gamma_0$ is the outer boundary component and $\Gamma_1,\ldots,\Gamma_M$ surround the enclosed voids. The outward unit normal is denoted by $\bn$. For readability, the arguments are written for sufficiently smooth fields and boundaries.

\subsection{Harmonic vector fields}
According to whether the curl and divergence vanish, vector fields in electromagnetics can be divided into the following four basic types:
\begin{enumerate}
\item \emph{Irrotational and divergent fields}: $\curl\bF=0$ and $\diver\bF\ne0$. A typical example is the electrostatic electric field in a charged region.

\item \emph{Rotational and solenoidal fields}: $\curl\bF\ne0$ and $\diver\bF=0$. A typical example is the magnetostatic magnetic flux density in a current-carrying region.

\item \emph{Rotational and divergent fields}: $\curl\bF\ne0$ and $\diver\bF\ne0$. General sourced time-varying electromagnetic fields may belong to this type.

\item \emph{Irrotational and solenoidal fields}: $\curl\bF=0$ and $\diver\bF=0$. Such fields arise naturally in static Maxwell problems and electromagnetic DC modes.

\end{enumerate}

The fourth type is the main concern of this paper. In a bounded domain $\Omega$, boundary conditions are needed to define the corresponding harmonic field space. This space may contain nontrivial solutions, and its dimension depends on the topology of the domain.

A harmonic Neumann field is a solution of the following vector boundary value problem:
\begin{subnumcases}{\label{eq:HNproblem}}
\curl\bF=\mathbf 0, & $\text{in }\Omega$,\label{eq:HN1}\\
\diver\bF=0, & $\text{in }\Omega$,\label{eq:HN2}\\
\bn\cdot\bF=0, & $\text{on }\partial\Omega$.\label{eq:HN3}
\end{subnumcases}
Thus a harmonic Neumann field is tangent to the boundary. Since \eqref{eq:HNproblem} is a homogeneous linear system, all of its solutions form a linear space. This linear space is denoted by $\HN$, which is called the harmonic Neumann space.

Similarly, a harmonic Dirichlet field is a solution of
\begin{subnumcases}{\label{eq:HDproblem}}
\curl\bF=\mathbf 0, & $\text{in }\Omega$,\label{eq:HD1}\\
\diver\bF=0, & $\text{in }\Omega$,\label{eq:HD2}\\
\bn\times\bF=\mathbf 0, & $\text{on }\partial\Omega$.\label{eq:HD3}
\end{subnumcases}
Thus a harmonic Dirichlet field has zero tangential component on the boundary. Equation~\eqref{eq:HDproblem} is also a homogeneous linear system. Its solutions form a linear space, denoted by $\HD$, which is called the harmonic Dirichlet space.

Under the assumptions on $\Omega$ used in this paper, these two harmonic spaces are finite-dimensional. Their dimensions are not determined by the local differential equations alone. They depend on the global topology of $\Omega$. In particular, the dimensions of $\HN$ and $\HD$ are related to the first and second Betti numbers, respectively. Note that the first and second Betti numbers are important topological invariants. This relation will be given in the following sections.

\subsubsection{Geometric interpretation of the first Betti number}

The first Betti number $b_1(\Omega)$ is a nonnegative integer that measures the number of independent handle-type circulation channels of the domain. Geometrically, it may be interpreted as the number of independent handle-type circulation cycles in $\Omega$. These cycles may be represented by non-contractible closed loops, which cannot be continuously shrunk to a point while remaining inside the domain.

The word ``independent'' is important for the first Betti number. A loop that winds twice or three times around the same handle does not introduce a new topological degree of freedom. It merely repeats the same basic circulation. Thus, a solid ball has no handle and $b_1=0$. A solid torus has one independent handle and $b_1=1$. A double-handle body has two independent handles and $b_1=2$. In general, a body with $n$ independent handles has $b_1=n$, and so forth.

Assume that $g=b_1(\Omega)$. One can choose $g$ oriented closed loops
$
\gamma_1,\ldots,\gamma_g
$
to represent the independent handle-type circulations of $\Omega$. These loops are not unique. What matters is that together they contain all independent circulation information of the domain. More precisely, for any oriented closed loop $L$ in $\Omega$, there exist integers $m_1,\ldots,m_g$ such that
$
L-\sum_{i=1}^{g}m_i\gamma_i
$
is the boundary of an oriented surface chain contained in $\overline{\Omega}$. Here, \(m_i\) is an integer whose sign indicates the winding direction around the \(i\)-th independent handle.

This property is particularly useful for curl-free fields. If $\curl\bF=\mathbf 0$, Stokes' theorem applied to the above surface chain gives
$$
\oint_L\bF\cdot d\mathbf l
=
\sum_{i=1}^{g}m_i
\oint_{\gamma_i}\bF\cdot d\mathbf l.
$$
Hence the circulation of a curl-free field around an arbitrary closed loop is completely determined by its circulations around the $g$ basic loops. In particular, if the circulation vanishes around all $\gamma_i$, then it vanishes around every closed loop in $\Omega$. This observation will form the basis of Lemma~\ref{lem:period}.

The basic loops may be chosen in different locations without changing the number $g$. For the bounded connected three-dimensional domains considered here, a set of independent representatives can in particular be chosen on the physical boundary $
\gamma_1,\ldots,\gamma_g\subset\partial\Omega.
$
Alternatively, equivalent representatives may be taken in the interior. This freedom will be important later: when a zero tangential-field condition is imposed on $\partial\Omega$, choosing the basic loops on the boundary immediately determines their circulations.

The first Betti number also has a particularly simple geometric description in terms of the boundary. Let $g_k$ denote the genus of $\Gamma_k$. Then, for the bounded connected three-dimensional domains considered here,
\begin{equation*}
b_1(\Omega)=\sum_{k=0}^{M}g_k.
\end{equation*}
Thus, the first Betti number is simply the total number of handles carried by all connected boundary components. Handles belonging to internal void boundaries contribute to $b_1$ in exactly the same way as handles belonging to the outer boundary.

\subsubsection{Geometric interpretation of the second Betti number}
The second Betti number has an even simpler interpretation for a bounded connected region $\Omega$ in $\RR^3$:
\begin{equation*}
b_2(\Omega)=M,
\end{equation*}
where $M$ is the number of enclosed voids in \eqref{eq:boundary}. In other words, $b_2$ counts independent enclosed voids.

The word ``independent'' is also important for the second Betti number. Two separated enclosed voids contribute two independent topological degrees of freedom. If two hollow regions are connected and form a single enclosed void, they count as only one. Thus, for the class of bounded three-dimensional domains considered here, $b_2(\Omega)$ counts the independent connected enclosed voids of the domain.

Table~\ref{tab:betti} summarizes several common geometries. Here a ball can be replaced by any smooth block that is topologically equivalent to a ball.
\begin{table}[!t]
\caption{The First and Second Betti numbers of Several Common Geometries}
\label{tab:betti}
\centering
\renewcommand{\arraystretch}{1.15}
\begin{tabular}{p{0.48\columnwidth}cc}
\hline
Geometry $(\Omega)$& $b_1(\Omega)$ & $b_2(\Omega)$\\
\hline
Solid ball or ordinary solid block & 0 & 0\\
Solid torus & 1 & 0\\
Double-handle solid & 2 & 0\\
Ball with one spherical void & 0 & 1\\
Ball with two spherical voids & 0 & 2\\
Ball with one toroidal void & 1 & 1\\
Trefoil-shaped solid tube & 1 & 0\\
\hline
\end{tabular}
\end{table}
The last example is useful for separating Betti numbers from knotting. A trefoil-shaped solid tube is still topologically a solid torus, so $b_1=1$ and $b_2=0$. Betti numbers count independent handles and enclosed voids, and they do not distinguish how a handle is knotted in three-dimensional space.

Two observations are important in electromagnetics. First, a simply connected domain has $b_1=0$, because every closed loop can be continuously contracted to a point. However, simple connectivity does not imply $b_2=0$. For example, a spherical shell is simply connected but has one enclosed void and hence $b_2=1$. Second, the two harmonic-field spaces considered in this paper are determined by both the curl-free and divergence-free conditions and the corresponding boundary conditions. The Neumann harmonic space is trivial when $b_1=0$, whereas the Dirichlet harmonic space is trivial when $b_2=0$. Therefore, both spaces are trivial only when $b_1(\Omega)=b_2(\Omega)=0$.

\subsection{Helmholtz decomposition and the harmonic component}
In electromagnetics, a frequently used vector-calculus picture is
\begin{equation}\label{eq:naivehelm}
\bF=\grad\phi+\curl\bA,
\end{equation}
where the first term is irrotational and the second term is solenoidal. In a 3-D bounded domain, the field structure also depends on both the boundary conditions and the topology of the domain. A more complete schematic form is
\begin{equation}\label{eq:fullhelm}
\bF=\grad\phi+\curl\bA+\mathbf h,
\end{equation}
where $\mathbf h$ is a harmonic vector field satisfying
\begin{equation}\label{eq:harmonicgeneric}
\curl\mathbf h=\mathbf 0,
\qquad
\diver\mathbf h=0.
\end{equation}
The precise orthogonality and boundary conditions attached to the three terms depend on the chosen bounded-domain Helmholtz decomposition. The important physical point is simpler: the harmonic part can be nonzero when the region is topologically nontrivial.

The Lemma~2.2 in \cite{Hiptmair2002} gives two statements that are especially relevant here. In vector notation they may be read as follows. A curl-free field can be written in the form
\begin{equation}\label{eq:HiptN}
\curl\bF=0
\quad\Longrightarrow\quad
\bF=\grad\phi+\mathbf h_N,
\end{equation}
where $\phi\in H^{1}(\Omega)$, $\mathbf h_N\in\HN$ and $\dim\mathcal H_N(\Omega)=b_1(\Omega)$.

Similarly, for a curl-free field satisfying a homogeneous tangential boundary condition,
\begin{equation}\label{eq:HiptD}
\curl\bF=0,\quad \bn\times\bF=0
\quad\Longrightarrow\quad
\bF=\grad\phi_0+\mathbf h_D,
\end{equation}
where $\phi_0\in H_{0}^{1}(\Omega)$, $\mathbf h_D\in{\mathcal H_D(\Omega)}$ and $\dim\mathcal H_D(\Omega)=b_2(\Omega)$.

The abstract theory behind \eqref{eq:HiptN} and \eqref{eq:HiptD} is well developed. The goal below is to show why these formulas are geometrically natural and how they can be obtained with familiar electromagnetic tools.

\subsection{Zero-circulation criterion for a global scalar potential}
The following lemma gives a direct vector-calculus criterion for the existence of a globally single-valued scalar potential. Its proof uses circulation integrals and Stokes' theorem, without invoking abstract cohomology theory.

\begin{lemma}[Zero-circulation criterion]\label{lem:period}
Let $\bF$ be a continuously differentiable field satisfying
\begin{equation*}
\curl\bF=\mathbf 0\qquad\text{in }\Omega.
\end{equation*}
Let $\gamma_1,\ldots,\gamma_g$ be the independent handle loops and $g=b_{1}(\Omega)$. If
\begin{equation*}
\oint_{\gamma_i}\bF\cdot d\mathbf l=0,
\qquad i=1,\ldots,g,
\end{equation*}
then the circulation of $\bF$ around every closed curve in $\Omega$ is zero, and there exists a globally single-valued scalar potential $q$ such that $\bF=\grad q$.
\end{lemma}
\begin{proof}
Let $L$ be an arbitrary oriented closed curve in $\Omega$. In general, $L$ need not be the boundary of a surface contained in $\Omega$, so Stokes' theorem cannot be applied directly to $L$.

By the choice of the independent circulation loops
$\gamma_1,\ldots,\gamma_g$, there exist integers
$m_1,\ldots,m_g$ such that the cycle $
L-\sum_{i=1}^{g}m_i\gamma_i
$
bounds an oriented surface chain $S$ contained in $\overline{\Omega}$, with its interior lying in $\Omega$. Hence
$$
\partial S
=
L-\sum_{i=1}^{g}m_i\gamma_i.
$$
Stokes' theorem can therefore be applied to this boundary,
$$
\begin{aligned}
0
&=
\int_S(\curl\bF)\cdot\mathbf n\,dS\\
&=
\oint_{\partial S}\bF\cdot d\mathbf l=
\oint_L\bF\cdot d\mathbf l
-
\sum_{i=1}^{g}m_i
\oint_{\gamma_i}\bF\cdot d\mathbf l.
\end{aligned}
$$

By the zero-circulation assumption on the independent loops,
$\oint_{\gamma_i}\bF\cdot d\mathbf l=0,~i=1,\ldots,g,$
and consequently $\oint_L\bF\cdot d\mathbf l=0$. Since $L$ is arbitrary, the circulation of $\bF$ vanishes around every closed curve in $\Omega$. Fix $\mathbf x_0\in\Omega$ and define
$$
q(\mathbf x)
=
\int_{\mathbf x_0}^{\mathbf x}
\bF\cdot d\mathbf l.
$$

If two paths connect $\mathbf x_0$ to $\mathbf x$, their difference forms a closed curve. Since the circulation around every closed curve is zero, the above integral is independent of the chosen path. Therefore $q$ is globally single-valued, and differentiation gives
$$
\bF=\grad q
\qquad\text{in }\Omega.
$$
\end{proof}

Lemma~\ref{lem:period} separates the local differential condition of vanishing curl from the global topological obstruction caused by nonzero circulations around independent handles.

\begin{remark}
The statement ``curl-free implies a global scalar potential'' is generally false in a multiply connected domain. A curl-free field always admits a local potential, but a globally single-valued potential exists only when all independent circulations around the handles vanish. Lemma~\ref{lem:period} makes this condition explicit. For example, in a solid torus, a curl-free field may still have a nonzero circulation around the handle and therefore cannot be represented by a globally single-valued scalar potential. The circulation conditions remove this topological obstruction.
\end{remark}

\section{Harmonic Neumann and Dirichlet Fields}
\indent The topology of $\Omega$ gives rise to two different types of harmonic vector fields, depending on the boundary condition. Harmonic Neumann fields are associated with independent handle-type circulations and are therefore related to the first Betti number. Harmonic Dirichlet fields are associated with independent potential differences between disconnected boundary components and are related to the second Betti number. These two spaces are studied separately below.
\subsection{Harmonic Neumann fields and the first Betti number}
\indent The main topological feature associated with a harmonic Neumann field is a nonzero circulation around an independent handle. Such a circulation cannot be represented by the gradient of a globally single-valued scalar potential in the original domain. The standard way to recover a scalar-potential representation is to introduce suitable cutting surfaces. This leads to a constructive basis for the harmonic Neumann space.

\subsubsection{Cut-potential basis}

A handle permits a nonzero circulation even though the field is curl-free. To represent such a circulation by a scalar potential, introduce one cutting surface for each independent handle. Let $g=b_1(\Omega)$ and choose $g$ connected orientable cutting surfaces
\begin{equation}\label{eq:cuts}
\Sigma_1,\ldots,\Sigma_g,
\qquad
\partial\Sigma_j\subset\partial\Omega.
\end{equation}
The cutting surfaces are chosen so that the independent handle-type circulation obstructions are removed by the cuts. The basic loops $\gamma_i$ and the cutting surfaces $\Sigma_j$ are chosen consistently. More precisely, the oriented intersection number of $\gamma_i$ with $\Sigma_j$ is $\delta_{ij}$. Thus $\gamma_i$ crosses its corresponding cut $\Sigma_i$ once with positive orientation and has zero algebraic intersection with the other cuts.

For each $j=1,\ldots,g$, let $\psi_j$ be the solution, unique up to an additive constant, of
\begin{subequations}\label{eq:cutproblem}
\begin{align}
\nabla^2\psi_j&=0,
&&\text{in~}\Omega\setminus\Sigma_j,\label{eq:cut1}\\
\frac{\partial\psi_j}{\partial n}&=0,
&&\text{on~}\partial\Omega,\label{eq:cut2}\\
\jump{\psi_j}_{\Sigma_j}&=1,\label{eq:cut3}\\
\jump{\frac{\partial\psi_j}{\partial n_j}}_{\Sigma_j}&=0.\label{eq:cut4}
\end{align}
\end{subequations}
Here $n_j$ is a fixed unit normal on $\Sigma_j$, and the same orientation is used to define the two traces and their jumps. The additive constant may be fixed, for example, by requiring zero mean. Under the assumed regularity conditions, standard variational theory for the Laplace problem gives existence and uniqueness modulo constants. This cut-potential construction is classical in the analysis of harmonic Neumann fields in multiply connected domains \cite{Alonso2018}.

The potential $\psi_j$ is unique only up to an additive constant. This nonuniqueness does not affect its gradient, since any two such potentials differ by a constant and therefore have the same gradient. Moreover, the gradient has matching traces on the two sides of the cut. Hence it defines a single-valued vector field on the original domain $\Omega$. We denote this extended gradient by
\begin{equation}\label{eq:Nj}
\bN_j=\widetilde{\grad\psi_j},
\end{equation}
where the tilde indicates the extension of $\grad\psi_j$ from the cut domain to the original uncut domain. Since $\psi_j$ is harmonic and satisfies the homogeneous Neumann condition on the physical boundary,
$$
\curl\bN_j=0,\qquad
\diver\bN_j=0,\qquad
\bn\cdot\bN_j=0.
$$
Therefore $\bN_j\in\HN$.

The prescribed potential jump determines the circulation of $\bN_j$. Since the oriented intersection number of $\gamma_i$ with $\Sigma_j$ is $\delta_{ij}$, and the potential jump across $\Sigma_j$ is equal to one according to \eqref{eq:cut3}, the line integral of $\bN_j$ around $\gamma_i$ is
\begin{equation}\label{eq:periodmatrix}
\oint_{\gamma_i}\bN_j\cdot d\mathbf l
=\delta_{ij},
\qquad i,j=1,\ldots,g.
\end{equation}
Thus $\bN_j$ has unit circulation around the $j$th independent handle and zero circulation around all the other basic handles.

\subsubsection{Dimension of the Neumann harmonic space}
\begin{theorem}\label{thm:N}
For a bounded, orientable and connected three-dimensional domain of finite topology, $\dim\HN=b_1(\Omega)$.
\end{theorem}

\begin{proof}
Equation \eqref{eq:periodmatrix} immediately shows that the $g=b_1(\Omega)$ fields $\bN_1,\ldots,\bN_g$ are linearly independent.

It remains to show that there are no other linearly independent Neumann harmonic fields. Let $\bF\in\HN$ and measure its $g$ basic circulations, $\alpha_i=\oint_{\gamma_i}\bF\cdot d\mathbf l$. Define $\bG=\bF-\sum_{j=1}^{g}\alpha_j\bN_j$. Then $\bG\in\HN$. Using \eqref{eq:periodmatrix}, one obtains
\begin{equation}\label{eq:Gzero}
\oint_{\gamma_i}\bG\cdot d\mathbf l=0,
\qquad i=1,\ldots,g.
\end{equation}
By \eqref{eq:Gzero} and $\nabla\times\bG={\bf0}$, Lemma~\ref{lem:period} shows that
\begin{equation}\label{eq:Ggrad}
\bG=\grad q
\end{equation}
for a globally single-valued scalar potential $q$. From $\diver\bG=0$ and $\bn\cdot\bG=0$,
\begin{equation}\label{eq:qNeu}
\nabla^2q=0\quad\text{in }\Omega,
\qquad
\frac{\partial q}{\partial n}=0\quad\text{on }\partial\Omega.
\end{equation}
Green's first identity gives
\begin{equation}\label{eq:greenN}
\int_\Omega |\grad q|^2\,d\Omega
=
\int_{\partial\Omega}q\frac{\partial q}{\partial n}\,dS
-
\int_\Omega q\nabla^2q\,d\Omega
=0.
\end{equation}
Hence $\grad q=0$, so $\bG=0$. Consequently every field in $\HN$ is a linear combination of $\bN_1,\ldots,\bN_g$, hence $\dim\HN=b_1(\Omega)$ is valid.
\end{proof}

The physical meaning of Theorem~\ref{thm:N} is simple: every independent handle permits one independent curl-free, divergence-free circulation that remains tangent to the boundary.

\subsection{Harmonic Dirichlet Fields and the Second Betti Number}\label{partd1}

\subsubsection{Existence of a Global Scalar Potential}

Let $\bF\in\HD$ and $g=b_1(\Omega)$, and choose the independent handle-type loops $\gamma_1,\ldots,\gamma_g$ on $\partial\Omega$. Since $\bn\times\bF=0$ on $\partial\Omega$, $\bF$ has no tangential component along these loops. Hence
\begin{equation}\label{eq:Dperiod}
\oint_{\gamma_i}\bF\cdot d\mathbf l=0,
\qquad i=1,\ldots,g.
\end{equation}
According to \eqref{eq:Dperiod} and $\curl\bF=0$, Lemma~\ref{lem:period} gives a globally single-valued scalar potential
\begin{equation}\label{eq:Dpot}
\bF=\grad\phi.
\end{equation}
Thus, the global potential follows from the curl-free condition together with the zero tangential boundary condition, not from the curl-free condition alone.

Since $\diver\bF=0$, $\phi$ satisfies $\nabla^2\phi=0$ in $\Omega$. Moreover, $\bn\times\grad\phi=0$ on $\partial\Omega$, so $\phi$ is constant on each connected boundary component:
\begin{equation}\label{eq:Ck}
\phi=C_k\quad\text{on }\Gamma_k,
\qquad k=0,1,\ldots,M.
\end{equation}
where $M=b_{2}(\Omega)$.
\subsubsection{One Field for Each Enclosed Void}

For each internal boundary component $\Gamma_k$, $k=1,\ldots,M$, let $\varphi_k$ solve
\begin{subnumcases}{\label{eq:hmeasure}}
\nabla^2\varphi_k=0 &$\text{in }\Omega$,\label{eq:hmeasure1}\\
\varphi_k=\delta_{kj} &$\text{on }\Gamma_j,\quad j=1,\ldots,M$,\label{eq:hmeasure2}\\
\varphi_k=0 &$\text{on }\Gamma_0$.\label{eq:hmeasure3}
\end{subnumcases}

Define $\bD_k=\grad\varphi_k$. Then $\bD_k$ is curl-free and divergence-free. Since every boundary component is an equipotential surface, $\bn\times\bD_k=0$ on $\partial\Omega$. Hence $\bD_k\in\HD$.

The fields $\bD_1,\ldots,\bD_M$ are linearly independent. Indeed, if $\sum_{k=1}^{M}c_k\bD_k={\bf0}$, then $\sum_{k=1}^{M}c_k\varphi_k$ is constant in $\Omega$. This constant is zero because the function vanishes on $\Gamma_0$. Evaluating it on $\Gamma_k$ gives $c_k=0$ for every $k$.

\begin{theorem}\label{thm:D}
For a bounded, orientable and connected three-dimensional domain of finite topology, $\dim\HD=b_2(\Omega)$.
\end{theorem}

\begin{proof}
Let $\bF\in\HD$. By \eqref{eq:Dpot} and \eqref{eq:Ck}, $\bF=\grad\phi$, where $\phi$ is harmonic in $\Omega$ and $\phi=C_k$ on $\Gamma_k$. Define
$$
q=\phi-C_0-\sum_{k=1}^{M}(C_k-C_0)\varphi_k.
$$
Then $\nabla^2q=0$ in $\Omega$ and $q=0$ on $\partial\Omega$. By the maximum principle for harmonic functions \cite{Evans2010}, $q=0$ in $\Omega$. Therefore,
$$
\bF=\sum_{k=1}^{M}(C_k-C_0)\bD_k.
$$
Hence $\bD_1,\ldots,\bD_M$ span $\HD$. Since they are also linearly independent, $\dim\HD=M=b_2(\Omega)$ holds.
\end{proof}

The physical meaning is direct. Each enclosed void allows one independent potential difference relative to the outer boundary. Therefore, each enclosed void contributes one independent harmonic Dirichlet field.

\section{Extensions and Electromagnetic Applications}
\subsection{Helmholtz decomposition revisited from an electromagnetic viewpoint}
The previous two sections allow the bounded-domain Helmholtz picture to be stated without the machinery of differential forms.

For the curl-free part of a vector field, topology gives the decomposition
\begin{equation}\label{eq:curlfreeN}
\curl\bF={\bf0}
\quad\Longrightarrow\quad
\bF=\grad\phi+
\sum_{j=1}^{b_1}\alpha_j\bN_j.
\end{equation}
The coefficients $\alpha_j$ are exactly the independent circulation integrals,
\begin{equation*}
\alpha_j=\oint_{\gamma_j}\bF\cdot d\mathbf l
\end{equation*}
when the basis is normalized as in \eqref{eq:periodmatrix}. Equation~\eqref{eq:curlfreeN} shows explicitly that a curl-free field is a global gradient if and only if all its independent circulations vanish.

If the curl-free field also satisfies the homogeneous tangential boundary condition, all independent handle circulations vanish. However, a different finite-dimensional freedom may remain because of enclosed voids. In this case,
\begin{equation}\label{eq:curlfreeD}
\curl\bF={\bf0},~~\bn\times\bF={\bf0}
~\Longrightarrow~
\bF=\grad\phi_0+
\sum_{k=1}^{b_2}\beta_k\bD_k,
\end{equation}
where $\phi_0$ may be chosen to vanish on the entire boundary, and the fields $\bD_k$ form a basis associated with the independent potential differences between the connected boundary components.

Equations \eqref{eq:curlfreeN} and \eqref{eq:curlfreeD} are the vector-calculus content behind the harmonic spaces appearing in \cite{Hiptmair2002}. They also explain the harmonic term $\mathbf h$ in \eqref{eq:fullhelm}. When $b_1(\Omega)=b_2(\Omega)=0$, there are no topological harmonic degrees of freedom, and the familiar textbook picture is recovered. When either Betti number is nonzero, dropping the harmonic component can remove physically admissible static fields from the decomposition.

This point is especially relevant to computational electromagnetics because the missing finite-dimensional spaces appear as null spaces or zero-frequency modes of discrete Maxwell operators. Edge-element methods preserve these topological structures when the discrete spaces are constructed consistently with the continuous Maxwell complex \cite{Hiptmair2002}.

\subsection{Extension to inhomogeneous anisotropic lossless media}

The above dimension formulas are topological and remain valid in an anisotropic lossless medium. Let $\boldsymbol\epsilon(\mathbf x)$ and $\boldsymbol\mu(\mathbf x)$ be the permittivity and permeability tensors of inhomogeneous, anisotropic and lossless media. Therefore, they are two bounded and uniformly Hermitian positive-definite tensors in $\Omega$.

For a PEC resonator, consider the source-free static electric-field problem
\begin{subnumcases}{\label{eq:HDC}}
\curl\bE=\mathbf 0, & $\text{in }\Omega$,\label{eq:Hdc1}\\
\diver(\boldsymbol\epsilon(\mathbf x)\bE)=0, & $\text{in }\Omega$,\label{eq:Hdc2}\\
\bn\times\bE=0, & $\text{on }\partial\Omega$.\label{eq:Hdc3}
\end{subnumcases}
The solutions of \eqref{eq:HDC} are referred to as generalized harmonic Dirichlet fields associated with the constitutive tensor $\boldsymbol\epsilon(\mathbf x)$. Since \eqref{eq:HDC} is a homogeneous linear system, its solutions form a linear space, denoted by $\mathcal H_{D,\epsilon}(\Omega)$.

The argument used in Theorem~\ref{thm:D} remains valid. In particular, \eqref{eq:Hdc1} and \eqref{eq:Hdc3}, together with Lemma~\ref{lem:period}, imply that every $\bE\in\mathcal H_{D,\epsilon}(\Omega)$ can be written as $\bE=\grad\phi$, where $\phi$ is constant on each connected boundary component.

For each internal boundary component $\Gamma_k$, $k=1,\ldots,M$, define $\varphi_k$ by
\begin{subnumcases}{\label{eq:hmeasures}}
\nabla\cdot(\boldsymbol\epsilon(\mathbf x)\nabla\varphi_k)=0 &$\text{in }\Omega$,\label{eq:hmeasures1}\\
\varphi_k=\delta_{kj} &$\text{on }\Gamma_j,~j=1,\ldots,M$,\label{eq:hmeasures2}\\
\varphi_k=0 &$\text{on }\Gamma_0$.\label{eq:hmeasures3}
\end{subnumcases}
Then $\bD_k^{\epsilon}=\grad\varphi_k$ satisfies \eqref{eq:HDC}. The same boundary-value argument used in Theorem~\ref{thm:D} shows that the fields
$\bD_1^{\epsilon},\ldots,\bD_M^{\epsilon}$ are linearly independent.

For an arbitrary $\bE=\grad\phi\in\mathcal H_{D,\epsilon}(\Omega)$, let $\phi=C_j$ on $\Gamma_j$ and define

$$
q=\phi-C_0-\sum_{k=1}^{M}(C_k-C_0)\varphi_k.
$$

Then $q=0$ on $\partial\Omega$ and
$\diver(\boldsymbol\epsilon(\mathbf x)\grad q)=0$ in $\Omega$. Green's identity gives
\begin{equation}\label{eq:epsenergy}
\int_\Omega\overline{\grad q}\cdot
(\boldsymbol\epsilon(\mathbf x)\,\grad q)\,d\Omega=0.
\end{equation}
Since $\boldsymbol\epsilon$ is uniformly Hermitian positive definite, \eqref{eq:epsenergy} implies $\grad q=\mathbf 0$. Hence
$$
\bE=
\sum_{k=1}^{M}(C_k-C_0)\bD_k^{\epsilon}.
$$
Therefore $\bD_1^{\epsilon},\ldots,\bD_M^{\epsilon}$ form a basis of the generalized harmonic Dirichlet space $\mathcal H_{D,\epsilon}(\Omega)$, and
\begin{equation}\label{eq:E0dim}
  \dim\mathcal H_{D,\epsilon}(\Omega)=M=b_2(\Omega).
\end{equation}

When the material is homogeneous, isotropic and lossless, $\mathcal H_{D,\epsilon}(\Omega)$ reduces to the usual harmonic Dirichlet space $\HD$. For a general uniformly Hermitian positive-definite tensor $\boldsymbol\epsilon(\mathbf x)$, the generalized harmonic Dirichlet fields depend on the material tensor, but the dimension of the corresponding space $\mathcal H_{D,\epsilon}(\Omega)$ remains the topological quantity $b_{2}(\Omega)$.

Similarly, consider the source-free static magnetic-field problem
\begin{subnumcases}{\label{eq:hhmc}}
\curl\bH=\mathbf 0, & $\text{in }\Omega$,\label{eq:hmc1}\\
\diver(\boldsymbol\mu(\mathbf x)\bH)=0, & $\text{in }\Omega$,\label{eq:hmc2}\\
\bn\cdot(\boldsymbol\mu(\mathbf x)\bH)=0, & $\text{on }\partial\Omega$.\label{eq:hmc3}
\end{subnumcases}
The solutions of \eqref{eq:hhmc} are referred to as generalized harmonic Neumann fields associated with the constitutive tensor $\boldsymbol\mu(\mathbf x)$. Since \eqref{eq:hhmc} is a homogeneous linear system, its solutions form a linear space, denoted by $\mathcal H_{N,\mu}(\Omega)$.

The construction is similar to that of the ordinary harmonic Neumann fields. Let $\Sigma_1,\ldots,\Sigma_g$, where $g=b_1(\Omega)$, be the cutting surfaces introduced previously. For each $j=1,\ldots,g$, let $\psi_j^{\mu}$ satisfy
\begin{subnumcases}{\label{eq:mucut}}
\diver(\boldsymbol\mu(\mathbf x)\grad\psi_j^{\mu})=0,
& $\text{in }\Omega\setminus\Sigma_j$,\label{eq:mucut1}\\
\bn\cdot(\boldsymbol\mu(\mathbf x)\grad\psi_j^{\mu})=0,
& $\text{on }\partial\Omega$,\label{eq:mucut2}\\
\jump{\psi_j^{\mu}}_{\Sigma_j}=1,\label{eq:mucut3}\\
\jump{\bn_j\cdot(\boldsymbol\mu(\mathbf x)\grad\psi_j^{\mu})}_{\Sigma_j}=0.\label{eq:mucut4}
\end{subnumcases}
Here $\bn_j$ is a fixed unit normal on $\Sigma_j$. As in the isotropic case, $\psi_j^{\mu}$ is unique up to an additive constant. The constant jump in \eqref{eq:mucut3} gives matching tangential traces of $\grad\psi_j^{\mu}$ across the cut, while \eqref{eq:mucut4} gives matching normal fluxes of $\boldsymbol\mu(\mathbf x)\grad\psi_j^{\mu}$. Hence the gradient extends from the cut domain to a field on the original domain that satisfies the generalized harmonic equations. Denote this extended field by
\begin{equation}\label{eq:Nmuj}
\bN_j^{\mu}=\widetilde{\grad\psi_j^{\mu}}.
\end{equation}
It follows from \eqref{eq:mucut} that $\bN_j^{\mu}\in\mathcal H_{N,\mu}(\Omega)$. Moreover, the unit jump in \eqref{eq:mucut3} and the dual relation between $\gamma_i$ and $\Sigma_j$ give
\begin{equation}\label{eq:muperiod}
\oint_{\gamma_i}\bN_j^{\mu}\cdot d\mathbf l
=\delta_{ij}.
\end{equation}
Hence $\bN_1^{\mu},\ldots,\bN_g^{\mu}$ are linearly independent.

It remains to show that these fields span $\mathcal H_{N,\mu}(\Omega)$. Let $\bH\in\mathcal H_{N,\mu}(\Omega)$ and define
$$
\alpha_i=\oint_{\gamma_i}\bH\cdot d\mathbf l,
\qquad
\bG=\bH-\sum_{j=1}^{g}\alpha_j\bN_j^{\mu}.
$$

By \eqref{eq:muperiod}, all independent circulations of $\bG$ vanish. Since $\curl\bG=\mathbf0$, Lemma~\ref{lem:period} gives $\bG=\grad q$ for a globally single-valued scalar potential $q$. Equations \eqref{eq:hmc2} and \eqref{eq:hmc3} then give
$$
\diver(\boldsymbol\mu(\mathbf x)\grad q)=0
\quad\text{in }\Omega,
\qquad
\bn\cdot(\boldsymbol\mu(\mathbf x)\grad q)=0
\quad\text{on }\partial\Omega.
$$

Green's identity yields
\begin{equation}\label{eq:muenergy}
\int_\Omega
\overline{\grad q}\cdot
\boldsymbol\mu(\mathbf x)\grad q\,d\Omega=0.
\end{equation}
Since $\boldsymbol\mu$ is uniformly Hermitian positive definite, \eqref{eq:muenergy} implies $\grad q=\mathbf0$. Thus $\bG=\mathbf0$, and every field in $\mathcal H_{N,\mu}(\Omega)$ is a linear combination of $\bN_1^{\mu},\ldots,\bN_g^{\mu}$. Therefore
\begin{equation}\label{eq:H0dim}
\dim\mathcal H_{N,\mu}(\Omega)=b_1(\Omega).
\end{equation}

When the material is homogeneous, isotropic and lossless, $\mathcal H_{N,\mu}(\Omega)$ reduces to the usual harmonic Neumann space $\HN$. For a general uniformly Hermitian positive-definite tensor $\boldsymbol\mu(\mathbf x)$, the generalized harmonic Neumann fields depend on the material tensor, but the dimension of the corresponding space $\mathcal H_{N,\mu}(\Omega)$ remains the topological quantity $b_{1}(\Omega)$.

Equations~\eqref{eq:E0dim} and \eqref{eq:H0dim} give a direct topological interpretation of the physical DC modes in a PEC resonator:
\begin{equation*}
\begin{aligned}
\dim\mathcal H_{D,\epsilon}(\Omega)=b_2(\Omega),\quad
\dim\mathcal H_{N,\mu}(\Omega)=b_1(\Omega).
\end{aligned}
\end{equation*}
Thus, the number of independent electric DC modes is equal to the number of enclosed voids, whereas the number of independent magnetic DC modes is equal to the number of independent handles.

If the PEC boundary has $M+1$ connected components, then $b_2(\Omega)=M$. Hence the electric DC mode space has dimension $M$, which is one less than the number of connected PEC boundary components. In contrast, the magnetic DC mode space is determined by the handle structure of the domain. Therefore, even when the boundary is connected, magnetic DC modes may still exist if the domain has nontrivial handles.

For real symmetric positive-definite constitutive tensors, the above energy arguments use the usual real inner product. For complex Hermitian positive-definite tensors, the corresponding complex inner product is used. In both cases, the constitutive tensors affect the forms of the generalized harmonic fields, but not the dimensions of the corresponding harmonic-field spaces.

\subsection{Relation to finite-element null spaces}

The above dimension formulas can also be examined numerically by finite-element discretization. For PEC resonators, edge-element formulations can be used to identify the physical DC subspace and to separate it from the ordinary gradient null space. The dimensions of the electric and magnetic DC spaces should agree with $b_2(\Omega)$ and $b_1(\Omega)$, respectively. Detailed finite-element formulations and numerical studies for topologically complex PEC resonators have been presented in \cite{Jiang2019} and are not repeated here.

\section{Conclusion}
A geometry-oriented vector-calculus explanation of Betti numbers and harmonic electromagnetic fields has been presented for bounded connected three-dimensional domains. The first Betti number $b_1$ is interpreted as the number of topologically independent handle-type circulation paths, whereas the second Betti number $b_2$ is the number of independent enclosed voids. Using only cutting surfaces, circulation integrals, scalar Laplace problems, Stokes' theorem, and Green's identity, the harmonic-field dimensions were obtained constructively as
\begin{equation*}
\dim\mathcal H_N(\Omega)=b_1(\Omega),
\qquad
\dim\mathcal H_D(\Omega)=b_2(\Omega).
\end{equation*}
These finite-dimensional spaces constitute the topological harmonic part that must be retained in bounded-domain Helmholtz decompositions. The same dimension formulas persist for uniformly positive-definite anisotropic media when the ordinary harmonic fields are replaced by the corresponding generalized harmonic fields. In electromagnetic resonators, these topological dimensions determine the multiplicities of the physical electric and magnetic DC modes. The results also explain the finite-dimensional null spaces that appear in edge-element discretizations of topologically complex resonators.

\section*{Acknowledgment}
This work is supported by the National Natural Science Foundation of China under Grant 61901131.

\end{document}